\documentclass{amsart}
\usepackage{amsfonts}

\usepackage{amsthm}
\usepackage{graphicx}
\usepackage[all]{xy}
\usepackage{color}

\DeclareGraphicsExtensions{.pdf,.eps,.fig,.pdf_tex}

\theoremstyle{plain}
\newtheorem{proposition}{Proposition}
\newtheorem{theorem}[proposition]{Theorem}
\newtheorem{lemma}[proposition]{Lemma}

\theoremstyle{definition}
\newtheorem{definition}{Definition}

\newtheorem{example}{Example}
\newtheorem{remark}{Remark}

\title{A Functorial Construction of Quantum Subtheories}
\author{Ivan Contreras
\and
Ali Nabi Duman
}
\address{Department of Mathematics\\
University of Illinois at Urbana-Champain\\
Illinois, USA}
\address{Department of Mathematics and Statistics\\
King Fahd University of Petroleum and Minerals\\
Dhahran, Saudi Arabia}
\email{icontrer@illinois.edu, aliduman@kfupm.edu.sa}

\begin{document}
\begin{abstract}
We apply the geometric quantization procedure via symplectic groupoids proposed by  E. Hawkins to the setting of epistemically restricted toy theories formalized by Spekkens \cite{S}. In the continuous degrees of freedom, this produces the algebraic structure of quadrature quantum subtheories. In the odd-prime finite degrees of freedom, we obtain a functor from the Frobenius algebra in \textbf{Rel} of the toy theories to the Frobenius algebra of stabilizer quantum mechanics.
\end{abstract}
\maketitle
%\tableofcontents

\section{Introduction}

The aim of geometric quantization is to construct, using the geometry of the classical system, a Hilbert space and a set of operators on that Hilbert space which give the quantum mechanical
analogue of the classical mechanical system modeled by a symplectic manifold \cite{AE, BW, W}. Starting with a symplectic space $M$ corresponding to the classical phase space, the square integrable functions over $M$ is the first Hilbert space in the construction, called \emph{prequantum Hilbert space}. In this case, the classical observables are mapped to the operators on this Hilbert space and Poisson bracket is mapped to the commutator. The desired \emph{quantum Hilbert space} consists of the sections of the prequantum Hilbert space which the depends on the "position" variables. These "position" variables are obtained by splitting the phase space via the \emph{polarization} $P$ which is the lagrangian subspace (i.e. the maximal subspace where the symplectic form vanishes) of the phase space.

The space of functions on $M$ is a commutative algebra under the operations of pointwise addition and multiplication. A bivector field on $M$ determines a Poisson bracket so that $M$ can be regarded as an approximation to a noncommutative algebra. The quantization approach due to Riefel aims to obtain such a $C^*$-algebra which is approximated by the Poisson algebra of the functions on $M$ \cite{R}. In this case, the algebra after quantization is a continuous field of $C^*$-algebras rather than a single algebra. On the other hand, Hawkins suggests a quantization recipe using symplectic groupoids to obtain a single $C^*$-algebra \cite{H}. In this paper, we use the quantization formulation of Hawkins to investigate the epistemic toy theory due to Spekkens \cite{S, S1}.

Spekkens introduced this toy theory in support of epistemic view of quantum mechanics \cite{S1}. The toy theory reproduces a large part quantum theory by positing restrictions on the knowledge of an observer. The distinctively quantum phenomena arises in the toy theory include complementarity, no-cloning, no-broadcasting, teleportation, entanglement, Choi-Jamiolkowski isomorphism, Naimark extension etc. On the other hand, the phenomena, such as Bell inequality violations, non-contextuality inequality violations and computational speed-up, do not arise in the toy theory.

The toy theory, that we are interested in, is the generalization of the original theory to the continuous and finite variables \cite{S}. This is achieved by positing a restriction on what kind of statistical distributions over the space of physical states can be prepared. The new theory is called \emph{epistricted theory}. In this way, quantum subtheories, Gaussian subtheory of quantum mechanics, stabilizer subtheory for qutrits, Gaussian epistricted optics can be obtained from statistical classical theories, Liouville mechanics, statistical theory of trits, statistical optics; respectively.

The epistemic restriction defined on the classical phase-space states that an agent knows the values of a set of variables that commute relative to the Poisson bracket and maximally ignorant otherwise. Hence, a symplectic structure, that appears in the function space of the phase space, has mathematical correspondence with the ingredients of the quantization scheme. As a result, we conclude that the geometric quantization, via Hawkins’ symplectic groupoid approach, produces a $C^*$-algebra that encodes the algebraic structure of the quadrature subtheories. Moreover, this construction gives us a functor from epistricted theories to the quantum subtheories.

In the second part of this paper, we construct a similar quantization functor of the toy theory for discrete degrees of freedom. In this case, the toy theory is defined precisely same as the continuous case except that the finite dimensional symplectic vector space is over a finite field with odd prime characteristic. However, in order to apply groupoid quantization we resort the methods of Categorical Quantum Mechanics pioneered by Abramsky and Coecke \cite{AC}. The categorical description of the toy theory is given in \cite{CE, CES}, where toy theory is formulated as a subcategory of dagger compact symmetric monodial category of finite sets and relations \textbf{Rel} and the toy observables correspond to dagger frobenius algebras.

We start our construction with the dagger Frobenius algebras of the toy observables which are functorially characterized as groupoids by Heunen, Catteneo and the first author in \cite{CCH}. After equipping the resulting groupoid with a symplectic structure, we construct the pair groupoid to apply the quantization recipe of Hawkins. One can also obtain this pair groupoid from a different direction called \emph{$CP^*$-construction} introduced in \cite{CHK}. In the category of Hilbert spaces, Frobenius algebras  correspond to finite dimensional $C^*$-algebras under this construction as as consequence of \cite{V}. For the category \textbf{Rel}, the pair groupoids are the objects of $CP^*[\textbf{Rel}]$. Hence, our main result establishes a functor from the dagger Frobenius algebra in \textbf{Rel} for epistricted theories to the Frobenius algebra in the category of Hilbert spaces.

The outline of this paper is as follows. We begin section 2 with a brief summary of the geometric quantization procedure. We then discuss epistricted theories of continuous variables and their correspondence in geometric quantization framework. We next briefly review Eli Hawkins' groupoid quantization recipe from which we obtain the usual Moyal quantization as a twisted group $C^*$-algebra from the geometric formulation of epistricted theories. We finally conclude that the resulting $C^*$-algebra contains phase-space formalism for quadrature subtheories. In section 3, we follow the same quantization procedure for the odd-discrete degrees of freedom. We end the paper with the conclusion and discussions.

\section{Continuous degrees of freedom}

The main idea in this section is to describe the general framework of geometric quantization in the context of epistemically restricted theories with continuous variables. We start with a quick overview of the stantard literature on geometric quantization and then we move on to the interpretation for epistemically restricted theories. We end the section with the algebraic counterpart of geometric quantization, introduction Hawkins' approach of quantization via symplectic groupoids. The outcome of this approach is a $C^*$-algebra for the epistricted theory.

\subsection{Overview on Geometric quantization}
There are several ideas behind the construction of geometric quantization, however, the main objective is to produce quantum objects by using the geometry of the objects from the classical theory. In the sequel we follow closely the approach of Bates and Weinstein \cite{BW}.

\subsubsection{The WKB method}
The WKB picture appears as an effort to describe quantum mechanics from a geometric viewpoint. It essentially approximates the solution of the time-independent Schr\"odinger equation, in the form
$$\phi=e^{iS/\hbar},$$
where $S$ is a solution of the Hamilton-Jacobi equation
$$H(x,\partial S/\partial x)=E.$$

%Construction of quantum mechanics from the geometry of classical mechanical objects is the main focus of the geometric guantization. Here we follow the flow of the lecture notes of Bates and Weinstein \cite{BW} as they start with a specific case of WKB method and generalize this system to more general phase spaces. Finally, one can get the connection to algebraic quantization (deformation quantization) via symplectic groupoid quantization. Thus, our final aim is to obtain a $C^*$-algebra structure which contains the epistricted theory.

%On the other hand, the basic WKB picture stated in a symplectic geometric formalism is sufficient to cover the epistricted theories. WKB approximate solution of time-independent Schrodinger equation is  where $S$ is a function satisfying Hamilton-Jacobi equation $H(x,\partial S/\partial x)=E$.
We can then use the geometry of the phase space to realize the solution to the Schr\"odinger equation as a Lagrangian submanifold $\mathcal L$ of the level set $H^{-1}(E)$. More precisely, let us consider the semiclassical approximation for $\phi$. From the transport equation

%This solution can be realized as the Lagrangian sub-manifold of the level set $H^{-1}(E)$ For a "semi-classical" approximation, one considers the transport equation

$$a\triangle S+2\sum \frac{\partial a}{ \partial x_j}\frac{\partial S}{\partial x_j}=0 $$
where $a$ is a function on $\mathbb{R}^n$, and after multiplying both sides by $a$, we obtain that
\begin{equation}\label{Div}
\mbox{div} (a^2\nabla S)=0.
\end{equation}
Now, if we consider the vector field
$$X_{H|L}=\sum_{j}\frac{\partial S}{ \partial x_j}\frac{\partial}{\partial q_j}-\frac{\partial V}{ \partial q_j}\frac{\partial}{\partial p_j}$$ onto $\mathbb{R}^n$ where the hamiltonian $H$ is $H(q,p)=\sum p_i^2/2+V(q)$ and $|dx|=|dx_1\wedge \ldots \wedge dx_n|$ is the canonical density on $\mathbb{R}^n$, its projection $X^{(x)}_H$  onto $\mathbb R^n$ satisfies the following invariance condition
\begin{equation}\label{Lie}
\mathfrak{L}_{X^{(x)}_H}(a^2|dx|)=0,
\end{equation}
where $\mathfrak L$ denotes the Lie derivative, if we restrict to the Lagrangian submanifold $\mathcal L=im (dS)$. Since the vector field $X_H$ is tangent to the manifold $\mathcal L$ and $\mathfrak L$ is diffeomorphims invariant, Equation \ref{Lie} implies that the pullback
$\pi^*(a^2|dx|)$ is invariant under the flow of $X_H$, where $\pi:T^*\mathbb{R}^n\rightarrow \mathcal L$ denotes the projection onto $\mathcal L$.

%By considering this condition on the Lagrangian manifold $\mathcal L=im(dS)$, we can deduce $\mathfrak{L}_{X^{(x)}_H}(a^2|dx|)=0$, where $X^{(x)}_H$ is the projection of the vector field

%Since $X_H$ is tangent to $L$ by the Jacobi-Hamilton theorem and since lie derivative is invariant under diffeomorphism, the equation $\mathfrak{L}_{X^{(x)}_H}(a^2|dx|)=0$ implies that the pullback $\pi^*(a^2|dx|)$ is invariant under flow $X_H$,

This discussion implies that a semi-classical state can be defined geometrically as a Lagrangian submanifold $\mathcal L$ of $\mathbb{R}^{2n}$, equipped with a half density function $a$. This semi-classical state is stationary when $\mathcal L$ lies in the level set of the Hamiltonian and the half density $a$ is invariant under the Hamiltonian flow. Transformations of the state correspond to Hamiltonians on $\mathbb R^{2n}$.
To summarize this geometric picture, Table I exhibits the correspondence between semi-classical objects (of geometric nature) and quantum objects (of algebraic nature) in this particular case.

\subsubsection{Basic symplectic and Poisson geometry}
From now one, we consider finite dimensional vector spaces $V$ to be symplectic, if they are equipped with a non-degenerate skew form $\omega$. For a vector subspace $W$ of $V$, its orthogonal complement is defined by $W^{\bot}=\{x\in V: \omega(x,y)=0, \forall y\in W \}$. We have the following special cases for $W$:
\begin{itemize}
\item $W$ is isotropic if $W \subseteq W^{\bot}$.
\item $W$ is coisotropic if $W^{\bot} \subseteq W$.
\item $W$ is symplectic if $W\cap W^{\bot}=\{0\}$.
\item $W$ is Lagrangian if $W=W^{\bot}$.
\end{itemize}
It can be easily checked that if $W$ is Lagrangian, then
$\dim W= \frac{1}{2}\dim V $.\\
\begin{definition}
A manifold is called Lagrangian (resp. isotropic, coisotropic and symplectic) if its tangent space is a Lagrangian subspace at every point.
\end{definition}
We also consider Poisson algebras, which are commutative algebras $(P,+,\bullet)$ equipped with a Lie bracket $[,]$ that is a derivation for the commutative product. As a particular case in our discussion, the algebra of functions of a symplectic manifold $(M,\omega)$ is naturally a Poisson algebra.

\begin{table*}[t]
  \centering
\begin{tabular}{|l|p{6cm}|p{6cm}|}
  \hline
  % after \\: \hline or \cline{col1-col2} \cline{col3-col4} ...
  Object & Semi-classical (geometric) version& Quantum (algebraic) version\\ \hline \hline
  phase space & $(\mathbb{R}^{2n},\omega)$ & Hilbert space $\mathfrak{H}_{\mathbb{R}^{2n}}$ \\ \hline
  state & Lagrangian submanifold of $\mathbb{R}^{2n}$ with half-density & half-density on $\mathbb{R}^{n}$ \\ \hline
  transformations & Hamiltonian $H$ on $\mathbb{R}^{2n}$ & operator $\hat{H}$ on smooth half densities \\ \hline
  stationary state & Lagrangian submanifold in level set of $H$  with invariant half-density & eigenvector of $\hat{H}$ \\
  \hline
\end{tabular}
\caption{Correspondence between classical and quantum objects}
  \label{tab:1}
\end{table*}

\subsubsection{Prequantum line bundle}
In this section, we follow Dirac's approach to axiomatize the quantization procedure.
\begin{definition}
A prequantization is a linear map $P\to \hat P_{\mathcal H}$ from a Poisson algebra (more precisely the algebra of functions of a Poisson manifold $M$) into the set of operators on a (pre)-Hilbert space $\mathcal H$, satisfying the following properties:
\begin{enumerate}
\item $\hat{\mbox{Id}_{P}}=\mbox{Id}_{\hat P_{\mathcal H}}$.
\item $\hat {[F,G]}=\frac i h (\hat F \hat G -\hat G \hat F)$.
\item $\hat{F^{*}}=(\hat F)^*$, where $*$ denotes complex conjugation on left side, and adjunction on the right side.
\end{enumerate}
\end{definition}
\begin{definition}
A prequantization is called quantization if, in addition to the properties above, the following condition is satisfied:
\begin{itemize}
\item 4. For a complete set of functions $\{F_i\}$, its quantization $\{\hat F_i\}$ is also a complete set of operators.

\end{itemize}
\end{definition}

\begin{proposition}
In the specific case where $M$ is a cotangent bundle $T^*N$, a prequantization \footnote{refered in the literature as the Koopman-Van Hove- Segal prequantization} can be constructed and it has the following form
\begin{equation}\label{PQ}
\hat F= F + \frac{\hbar}{2\pi i}X_F-\theta(X_F),
\end{equation}
where $X_F$ is a Hamiltonian vector field with generating function $F$ and $\theta$ is a primitive of the Liouville form $\omega_{T^*N}$.
\end{proposition}
In order to implement this prequantization for a  arbitrary symplectic manifold $(M, \omega)$, we require a complex line bundle over $M$, equipped with a Hermitian structure and a Hermitian connection $\nabla$, for which the prequantization formula \ref{PQ} takes the following form
\begin{equation}
\hat F= F +\frac{\hbar}{2\pi i} \nabla_{X_F}.
\end{equation}

Provided a compatibility condition between curv$(\nabla)$ and $\omega$, this formula gives a prequantization for $(M,\omega)$.

\subsubsection{Polarization}
It is easy to realize in some examples that the Hilbert space of prequantization is too big for the completeness condition 4 to hold. By using the ordinary viewpoint of quantum mechanics, only half of the coordinates of the classical phase space are required to write down the wave functions, depending whether the coordinate or momentum representation is considered. In (symplectic) geometric terms, for general symplectic manifolds, a polarization is defined as follows:
\begin{definition} Let $(M, \omega)$ be a symplectic manifold. A \emph{polarization} of $M$ is a Lagrangian involutive distribution $\mathcal P$ of $M$.
\end{definition}

Thus the quantization space consists of functions constant along the leaves of a the distribution $\mathcal P$ on $M$, more precisely, the quantization Hilbert sapce $\mathcal H$ is the space of sections $s$ of the complex line bundle on $M$ such that
\begin{equation}
\nabla_{X_P}s=0,
\end{equation}
where $X_P$ is a vector field tangent to the polarization $\mathcal P$.

\subsection{Quadrature Epistricted Theories}

We now introduce the quadrature epistricted theories for continuous variables \cite{S}. The epistemic restrictions on classical variables are adopted from the condition of the joint measurability of quantum observables. 
\begin{definition}A set of variables are \emph{jointly knowable} if and only if it is commuting with respect to the Poisson bracket. 
\end{definition}
The other restriction besides joint knowability is that an agent can know only the variables which are linear combination of the position and momentum variables.
\begin{example}\textbf{(Darboux coordinates).}\label{Dar}
If we start with the phase space $\Omega=\mathbb{R}^{2n}$ where a point is denoted by $\mathbf{m}=(\mathbf{p}_1,\mathbf{q}_1,\ldots,\mathbf{p}_n,\mathbf{q}_n)$, epistemic restrictions imply that the functionals $f:\Omega \to \mathbb{R}$ are of the form $$f=\mathbf{a_1}q_1+\mathbf{b_1}p_1+\ldots + \mathbf{a_n}q_n + \mathbf{b_n}p_n+ \mathbf{c}$$ where $\mathbf{a_1}, \mathbf{b_1},\ldots, \mathbf{a_n},\mathbf{b_n},\mathbf{c}\in \mathbb{R}$ and $p_i(\mathbf{m})=\mathbf{p_i}$ and $q_i(\mathbf{m})=\mathbf{q_i}$ are functionals associated with momentum and position , respectively. Hence, each functional $f$ is associated with a vector $\mathbf{f}=(\mathbf{a_1}, \mathbf{b_1},\ldots, \mathbf{a_n},\mathbf{b_n})$. It is not hard to show that the value of the Poisson bracket over the phase space is uniform and equal to the symplectic inner product: $$[f,g]_{PB}(\mathbf{m})=\sum_{i=1}^n(\frac{\partial f}{\partial q_i}\frac{\partial g}{\partial p_i}- \frac{\partial g}{\partial q_i}\frac{\partial f}{\partial p_i})(\mathbf{m})=\langle \mathbf{f},\mathbf{g}\rangle$$
where
$$\langle \mathbf{f},\mathbf{g} \rangle=\mathbf{f}^T J \mathbf{g}$$
and $J$ is the skew symmetric $2n\times 2n$ matrix with components $J_{ij}=\delta_{i,j+1}-\delta_{i+1,j}$. Hence, the vector space $\Omega$ becomes a symplectic vector space with the symplectic inner product $\omega = \langle \cdot, \cdot \rangle$. This allows us to give the geometric presentation of the quadrature variables.

The only set of variables jointly knowable are the ones that are Poisson commuting. In symplectic geometry, this set corresponds to the subspace $V$ of vectors whose symplectic inner product vanish, i.e. $\forall \mathbf{f}, \mathbf{g}\in V$  $ \langle \mathbf{f}, \mathbf{g} \rangle = 0 $. For a $2n$-dimensional phase space, the maximum possible dimension of such a $V$ is $n$. Such a maximal space is a Lagrangian space as defined above and it corresponds to the maximal possible knowledge an agent can have. In order to specify an epistemic state one should also set the values of the variables on $V$. The linear functional $v$ acting on a quadrature functional corresponds to the set of vectors in $\mathbf{v}\in V$ which is determined via $v(f)=\mathbf{f}^T \mathbf{v} $. That is,for every vector $\mathbf{v}\in V$ we obtain distinct value assignment.

In summary, a pure state in the epistricted theory consists of a Lagrangian subspace $V \in \mathbb{R}^{2n} $ and a valuation functional $v: \mathbb{R}^{2n} \to \mathbb{R}$. In geometric quantization, the half density function can be regarded as this valuation function.

On the other hand, the valid transformations are the symplectic transformations which maps the quadrature variables to itself. These transformations map a phase space vector $\mathbf{m}\in \Omega$ to $S\mathbf{m} + \mathbf{a}$ where $\mathbf{a}$ is a displacement vector and $S$ is $2n\times 2n $ symplectic matrix. The group formed by these transformations is called the \emph{affine symplectic group}, which is subgroup of the hamiltonian symplectomorphism group. Thus, each of these transformations can be obtained from a hamiltonian. Finally, the sharp measurements are parametrized by Poisson commuting sets of quadrature variables (isotropic subspaces $V$) and the outcomes are indexed by the vectors in $V$.
\end{example}

We summarize the correspondence between geometric quantization and epistricted theories in Table 2.

\begin{table*}[t]
  \centering
\begin{tabular}{|l|p{6cm}|p{6cm}|}
  \hline
  % after \\: \hline or \cline{col1-col2} \cline{col3-col4} ...
  Object & Semi-classical version in quantization& Epistricted theories\\ \hline \hline
  phase space & $(\mathbb{R}^{2n},\omega)$ & $(\mathbb{R}^{2n},\omega)$ \\ \hline
  state & Lagrangian submanifold of $\mathbb{R}^{2n}$ with half-density $a:\mathbb{R}^{2n}\to \mathbb{R}$& Lagrangian subspace with a valuation function $v:\mathbb{R}^{2n}\to \mathbb{R}$ \\ \hline
  transformations & hamiltonian $H$ on $\mathbb{R}^{2n}$ & affine symplectic transformation \\
  \hline
\end{tabular}
\caption{Correspondence between geometric quantization and epistricted theories}
  \label{tab:1}
\end{table*}

%theorem 3.8

\subsection{Hawkins' Groupoid Quantization}
The aim of this section is to point out that the epistricted theories can be quantized by a twisted polarized convolution $C^*$-algebra of a symplectic groupoid in the sense of E. Hawkins. The main idea in this method is to find a $C^*$-algebra which is approximated by a Poisson algebra of functions on a manifold. $C^*$-algebra quantization is mainly developed by the work of Rieffel where the quantization is stated as a continuous field of $C^*$-algebras $\{\mathcal{A}_{\hbar}\}$. Hawkins' construction gives a single algebra $\mathcal{A}_1$ by involving additional structures on the symplectic groupoid. %In this way, one can also reconstruct many examples of geometric quantization unlike the other approaches  like the \emph{dictionary} strategy of Weinstein which can only be cover a few examples \cite{BW}.
In his approach, it is possible to reinterpret geometric quantization for a broader class of examples, coming from deformation quantization of Poisson algebras. This gives a rigorous treatment to the \emph{dictionary} strategy of Weinstein relating the symplectic category and its geometrically quantized counterpart \cite{BW}.
\subsubsection{Symplectic groupoids}
We start with the definition of symplectic groupoid, arising from the usual definition of Lie groupoid, requiring compatibility conditions with a symplectic structure on the space of arrows.
\begin{definition}
A \emph{topological groupoid} $\Sigma$
is a groupoid object in the category of topological spaces, that is, $\Sigma$ consists of a space of $\Sigma_0$ of objects and a space $\Sigma_2$ of arrows, together with five continous structure maps:
\begin{itemize}
\item The source map $s:\Sigma_2\rightarrow \Sigma_0$ assigns to each arrow $g\in \Sigma_2$ its source $s(g).$

\item The target map $t:\Sigma_2\rightarrow \Sigma_0$ assigns to each arrow $g\in \Sigma_2$ its target $t(g)$. For two objects $x$, $y\in \Sigma_0$, one writes $g:x\rightarrow y$ to indicate that $g\in \Sigma_2$ is an arrow with $s(g)=x$ and $t(g)$.

\item If $g$ and $h$ are arrows with $s(h)=t(g)$, one can form their composition, denoted $hg$, with $s(hg)=s(g)$ and $t(hg)=t(h)$. If $g:x\rightarrow y$ and $h:y\rightarrow z$ then $hg$ is defined and $hg:x\rightarrow z.$ The composition map $m$ is defined by $m(h,g)=hg$, and it is a well defined map $m: \Sigma_m \to \Sigma_2$, where $\Sigma_m:=\{(h,g): s(h)=t(g)\}.$

\item The unit map $u:\Sigma_0\rightarrow \Sigma_2$ is a two sided unit for composition.

\item The involution map $-^* :\Sigma_2\rightarrow \Sigma_2$. Here, if $g:x\rightarrow y$ then $g^*:y\rightarrow x$ is two sided inverse for composition.
\end{itemize}
\end{definition}

$\Sigma$ is said to be a groupoid over $\Sigma_0$

\begin{definition} A \emph{Lie groupoid} is a topological groupoid $\Sigma$ where $\Sigma_0$ and $\Sigma_2$ are smooth manifolds, and such that the structure maps $s$, $t$, $m$, $u$ and $-^*$ are smooth. Moreover, $s$ and $t$ are required to be submersions so that the domain of $m$ is a smooth manifold.
\end{definition}

\begin{definition} A Lie groupoid $\Sigma$ is called a \emph{symplectic groupoid} if $\Sigma_2$  is a symplectic manifold with symplectic form $\omega$ and the graph multiplication relation $\mathfrak{m}=\{(xy,x,y):(x,y)\in \Sigma_2\}$ is a Lagrangian submanifold of $\Sigma_2 \oplus \overline {\Sigma}_2 \oplus \overline{\Sigma}_2$, where $\overline{\Sigma}$ is the symplectic manifold $(\Sigma_2, -\omega)$.
\end{definition}
This definition is equivalent to say the the symplectic form $\omega$ is \emph{multiplicative}, i.e. it satisfies the following compatibility conditions with the multiplication and projection maps:
\begin{equation}
m^* \omega= pr_1^* \omega + pr_2^* \omega,
\end{equation}
where $pr_1$ and $pr_2$ are the projections of $\Sigma_m$ onto the first and second component, respectively.
As $\mathfrak{m}$ is Lagrangian, one can find a unique Poisson structure on $\Sigma_0$ of a symplectic groupoid such that $s$ is a Poisson map and $t$ is anti-Poisson. Hence, we have the following definition.

\begin{definition} A symplectic groupoid $\Sigma$ is said to \emph{integrate} a Poisson manifold $\Omega$ if there exists a Poisson isomorphism from $\Sigma_0$ onto $\Omega$.
\end{definition}
The following are basic examples of symplectic groupoids, the first one being of central importance for the geometric quantization procedure in epistricted theories.

\begin{example}
(Pair groupoid of a symplectic manifold).   \\
As we will describe more detailed in Definition \ref{pair}, given a smooth manifold $M$, the manifold $M\times M$ is naturally the space of arrows for a Lie groupoid, called the \emph{pair groupoid}. In the case where $M$ is equipped with a symplectic structure $\omega$, then the Lie groupoid Pair$(M)$ is a symplectic groupoid with simplectic structure $\omega \oplus \omega$
\end{example}

\begin{example}
(Cotangent bundle).\\
 If $M$ is a manifold, any vector bundle $E$ over $M$ is a Lie groupoid over $M$, the multiplication is given by fiber addition, the source and target maps map are projection onto the base, whereas the unit is given by the zero section of the bundle. In the particular case that $E=T^*M$ and that $\omega$ is the Liouville form on the cotangetn bundle, it it easy to verify that $T^*M$ is a symplectic groupoid over $M$.
\end{example}

 Here is the Hawkins' strategy for geometric quantization of a  manifold $\Omega$.For a detailed discussion, one can refer to \cite{H}.

\begin{itemize}
\item Construct an symplectic groupoid $\Sigma$ over $\Omega$.
\item Construct a prequantization $(\sigma, L, \nabla)$ of $\Sigma$.
\item Choose a symplectic groupoid polarization $P$ of $\Sigma$ which satisfies both symplectic and groupoid polarization.
\item Construct a "half form" bundle.
\item $\Omega$ is quantized by twisted, polarized convolution algebra $C^*_P(\Sigma, \sigma)$.

\end{itemize}

%In this work we only focus on the context of the epistricted theories for which $\Omega=\mathbb{R}^{2n}=\{(q_1,q_2,\cdots ,q_{2n}), q_i \in \mathbb R \}$, with symplectic form $\omega$.

\begin{proposition}\label{Moyal}
The Hawkins' geometric quantization of the symplectic space $\Omega=\mathbb{R}^{2n}$ and Darboux coordinates (Example \ref{Dar}) is the Moyal quantization of the Poisson algebra of the symplectic vector space.
\end{proposition}
\begin{proof}
In the particular case of symplectic manifold is a vector space $\Omega=\mathbb{R}^{2n}$ with symplectic form $\omega$, which is the context of the epistricted theories, we have the symplectic groupoid $\Omega \oplus \Omega^{*}$  integrating  the  symplectic vector space $\Omega$, where the multiplication is given by fiber addition on $\Omega^*=\{(p^1,p^2,\cdots , p^{2n})\}$, i.e. the symplectic integration comes equipped with Darboux coordinates.

More explicitly, $\hat{\omega}(u):v \mapsto \omega(u,v)$ gives a map $\hat{\omega}:\mathbb{R}^{2n} \to \mathbb{R}^{2n*}$. One obtains a symplectic structure
 $$\sigma((x,y),(z,w))=\omega(x,z)-\omega(y,w)$$ $$=\hat{\omega}(x-y)[\frac{z+w}{2}]-\hat{\omega}(z-w)[\frac{x+y}{2}].$$
 We identify $\mathbb{R}^{2n}\oplus \bar{\mathbb{R}}^{2n}$ with the cotangent bundle $T^*(\mathbb{R}^{2n})$ as follows: For the local coordinates of covectors $(u,\xi)$, $(v,\eta)$ in  $T^*(\mathbb{R}^{2n})$, the cotangent symplectic structure is $$\sigma^*((u,\xi),(v,\eta))=\xi(u)-\eta(v).$$ This gives us a symplectomorphism $\Phi:\mathbb{R}^{2n}\oplus \bar{\mathbb{R}}^{2n} \to T^*(\mathbb{R}^{2n})$ such that
 $$\Phi: (x,y)\mapsto (1/2(x+y),\hat{\omega}(x-y))$$ where $\Phi^*\sigma^*=\sigma$.  \footnote{ This example has also been studied by Hawkins (see example 6.2 \cite{H}).}

One can obtain the the Darboux coordinates $(q_1,\ldots, q_n, p_1,\ldots,p_n)$ of $T^*(\mathbb{R}^{2n})$ from the symplectomorphism $\Phi$. The projection of $T^*(\mathbb{R}^{2n})$ to $\mathbb{R}^{2n*}$ is a fibration of groupoids whose fibers are Lagragian. Thus this is a polarization of the symplectic groupoid given by
$$P=span\{\partial/\partial p_1,\ldots, \partial/\partial p_n \}$$
The symplectic potential which vanishes on $P$ can be chosen as $\theta_P= -p^i dq_i.$

This polarization gives us the half-form pairing, which enables quantizable observables to be represented as operators on the Hilbert space $L^2(\mathbb{R}^{2n})$.Hence, this yields the correspondence between the kernels of operators on $L^2(\mathbb{R}^{2n})$ and Weyl symbols of these operators. This kernel $T$ of an operator $f$ is given by
 $$Tf(p,q)=C\int f(\frac{p+q}{2},\zeta)e^{i\zeta (q-p)/\hbar} d\zeta.$$
%-----------------------------------------------------------------------------------------------

 The quantization procedure gives the twisted group algebra $C^*(\Omega^*, \sigma)$ where $\sigma:\Omega^* \times \Omega^* \rightarrow \mathbb{T}$, $\sigma(x, y)=e^{\frac{-i}{\{q,p \}}}$. This is the usual Moyal quantization of a Poisson vector space (see \cite{R2}).  In this setting, the observables corresponds to functions in classical phase-space and the Moyal product of functions is derived from the product of pair of observables.  In this case, the position and momentum operators correspond to the generators of the Heisenberg group and they are related to each other by a Fourier transform.
 \end{proof}
 
 \begin{theorem}
 Quadrature quantum subtheories and the Moyal quantization from Proposition \ref{Moyal} coincide.
 \end{theorem}
 \begin{proof}

To be consistent with the formalism of \cite{S}, we work with projector valued measures (PVM) rather than Hermitian operators. PVMs are used in quantum information and quantum foundations to represent measurements as eigenvalues of Hermitian operators are operationally insignificant and serve as labels of outcomes. A \emph{projector-valued measure} with outcome set $K$ is a set of projectors $\{\Pi_k: k\in K \}$ such that $\Pi_k^2=\Pi_k$, $\forall k \in K$ and $\sum_k \Pi_k=I$. Hence the position (momentum) observable are the set of projectors onto position (momentum) eigenstates\footnote{In the continuous case one can also use Hermitian operators corresponding to the real valued functionals but the commutation relation of Hermitian operators does not have finite counterpart Therefore, Spekkens preferred to use PVMs in order to cover finite and continuous cases simultaneously.}:
 $$\mathcal{O}_q=\{\hat{\Pi}_q(\mathbf{q}): \mathbf{q}\in \mathbb{R}\}$$
 where $$\hat{\Pi}_q(\mathbf{q})=| \mathbf{q}\rangle _q \langle \mathbf{q}|.  $$

We now define a unitary representation of symplectic affine transformation to introduce the other quadrature observables. The projective unitary representation $\hat{V}$ of the symplectic group acting on the phase space $\Omega$ satisfies $\hat{V}(S)\hat{V}(S')=e^{i\phi}\hat{V}(S S')$ for every symplectic matrix $S:\Omega \to \Omega$ and where $e^{i\phi}$ is a phase factor. The action of this unitary is defined by the conjugation
$$\mathcal{V}(S)(\cdot)=\hat{V}(S)(\cdot)\hat{V}^{\dag}(S).$$
For single degree of freedom, let $S_f$ be the symplectic matrix that takes the position functional $q$ to a quadrature functional $f$ such that $S_f\mathbf{q}=\mathbf{f}$. Then the \emph{quadrature observable} associated with $f$ is defined as follows
 $$\mathcal{O}_f=\{\hat{\Pi}_f(\mathbf{f}): \mathbf{f}\in \mathbb{R}\}$$
 where
 $$\hat{\Pi}_f(\mathbf{f})=\mathcal{V}(S_f) (\hat{\Pi}_q(\mathbf{f})).$$

 For the $n$ degrees of freedom $\Omega=\mathbb{R}^{2n}$, the quadrature observable associated with $f$ is given by $$\mathcal{O}_f=\{\hat{\Pi}_f(\mathbf{f}): \mathbf{f}\in \mathbb{R}^{2n} \}$$ where $$\hat{\Pi}_f(\mathbf{f})=\mathcal{V}(S_f) (I\otimes \cdots\otimes \hat{\Pi}_{q_i}(\mathbf{f})\otimes \cdots \otimes I)$$ for $S_f\mathbf{q}_i=\mathbf{f}$. We also know that the set of quadrature observables $\{\mathcal{O}_{f_i}\}$ commute if and only if the corresponding functionals $\{f_i\}$ are Poisson-commuting (see \cite{S1}). Hence, the commuting set of quadrature observables can be labelled by isotropic subspaces of $\Omega$. This set defines a single quadrature observable $$\mathcal{O}_{V}=\{\hat{\Pi}_V(\mathbf{v}):\mathbf{v}\in V\}$$ where $$\hat{\Pi}_V(\mathbf{v})=\prod_{\mathbf{f}^{(i)}}\hat{\Pi}_{f^{(i)}}(f^{(i)}\mathbf{v}).$$

 On the other hand, in the geometric quantization procedure, any functional $f$ on $\Omega$ is mapped to a hermitian operator $\hat{f}$ in a prequantum Hilbert space which corresponds to the observable $\mathcal{O}_f=\{\hat{\Pi}_f(\mathbf{f}): \mathbf{f}\in \mathbb{R}^{2n}\}$.  Moreover,  the commutation relation for the observables in both quadrature subtheories and geometric quantization, is implied by the Poisson commutation relation of the classical observables. As the polarization is the commuting set of these hermitian operators, the state that is obtained after quantization is the operator $\hat{\Pi}_V(\mathbf{v})$.The choice of the vertical polarization for the groupoid $\Omega \oplus \Omega^*$ is the responsible of the correspondence between the two quantum states. The half-form pairing defined above can be computed in terms of the integral kernel of the projection operator $\hat{\Pi}_f$, which has Weyl symbol $f$. This establishes a correspondence between phase-space formalism and quantum mechanics, and Moyal product is deduced from this correspondence.
 \end{proof}

 In \cite{S1}, the operational equivalence quantum subtheories and epistricted theories is proven using Wigner representation which maps operators in Hilbert space to the functions in phase-space formulation of quantum mechanics. It is also well-known fact that the Wigner representation of an operator product is given by the Moyal product. As a result, geometric quantization with an appropriate choice of polarization is operationally equivalent to epistricted theories. We can also conclude that group algebra $C^*(H)=C^*(\Omega^*,\sigma)$, which is the Hilbert space considered as a group representation of the Heisenberg group $H$, contains the algebraic structure of quadrature subtheories.

This discussion leads to the following theorem:
\begin{theorem}[Main result in the continuous case]
The geometric quantization, via Hawkins' symplectic groupoid approach, of the Spekkens toy theory of continuous degrees of freedom produces a $C^*$-algebra that is a group representation for the Heisenberg group $H$ and it encodes the algebraic structure of the quadrature subtheories, via Moyal quantization.
\end{theorem}

 %In the future, we intend to investigate the application operator algebraic approach to quantum optics.

\subsection{Functoriality}

The functoriality of geometric quantization is a delicate issue and it has been proven that the quantization that fits with the Schr\"oedinger picture is in fact not functorial. There are several problems even before quantization, in particular, that the symplectic category is not quite a category, since the composition of Lagrangian correspondences is not in general well defined, and also that when it is defined, the composition is not continuous with the standard topology in the Lagrangian Grassmanian. The failure of geometric quantization to functorially represent Schroedinger's picture is given e.g. in Gotay's work \cite{Go}.

However, the geometric quantization picture for symplectic groupoids turns out to be functorial with respect to the choices, i.e. the polarizations (the groupoid one), the half line bundle. The fact that the choices of polarizations are affine means that there is a higher structure for our C*-algebra quantization, namely, the objects are symplectic manifolds, 1-morphisms are Lagrangian polarizations and 2-morphisms are affine transformations between Lagrangian polarizations. These 2-morphims are reflected in C*-algebra automorphisms after quantization.

%affine sypmlectic group corresponds to nonprojective representation of affine symplecticgroup

%MORE DETAILS-------------------------------------------------------------------------------

\section{Finite degrees of freedom}

 We now discuss how the geometric quantization relates the epistricted theories to quadrature quantum subtheories for odd-prime discrete degrees of freedom. In \cite{S}, the operational equivalence of these two theories for continuous and odd-prime discrete cases was proven using Wigner representation. Here, we aim to construct a functor from a subcategory of the category of groupoids to the category of $C^*$-algebras. This corresponds to a functor from Frobenius algebras in the category \textbf{FRel} (Frobenius algebras in the category of sets and relations) to Frobenius algebras in the category of Hilbert spaces \textbf{FHilb}. Here is the sketch of our discrete quantization:
 \begin{itemize}
 \item We start with the special dagger Frobenius algebra of epistricted theories, \textbf{Spek}, which is a subcategory of finite sets and relations, \textbf{FRel}.
 \item We then construct the groupoid $\mathcal{G}$ corresponding to \textbf{Spek} via the explicit equivalence in Heunen et. al.\cite{CCH}.
 \item We next obtain the pair groupoid from $\mathcal{G}$ and introduce the symplectic structure on it which is compatible with the pair groupoid structure. In this case, each polarization corresponds to a Lagrangian subspace in epistricted theories.
 \item We then apply geometric quantization procedure a la Hawkins on the pair groupoid by considering the complex valued function space on the groupoid and using discrete Fourier transform (integral kernel) defined by Gross \cite{G}.
 \item Finally, we end up with the finite dimensional $C^*$-algebra from which one can construct special dagger Frobenius algebra over \textbf{FHilb} via \cite{V}.
 \end{itemize}

  We begin this section by reviewing the epistricted theories in discrete case.

\subsection{Quadrature Epistricted Theories}
The formalism in finite case is defined over the finite fields with prime order $d$. These fields are isomorphic to the integers modulo $d$, denoted by $\mathbb{Z}_d$. Hence, the configuration space and associated phase-space are $(\mathbb{Z}_d)^n$, $\Omega = (\mathbb{Z}_d)^{2n}$, respectively. The linear functionals are also in the form:
$$f=\mathbf{a_1}q_1+\mathbf{b_1}p_1+\ldots + \mathbf{a_n}q_n + \mathbf{b_n}p_n+ \mathbf{c}$$ where $\mathbf{a_1}, \mathbf{b_1},\ldots, \mathbf{a_n},\mathbf{b_n},c\in \mathbb{Z}_d$. Hence, a vector $\mathbf{f}=(\mathbf{a_1}, \mathbf{b_1},\ldots, \mathbf{a_n},\mathbf{b_n})$ specifies the position and momentum dependence of the quadrature functional $f$. The dual space $\Omega ^*= (\mathbb{Z}_d)^n$ consists of these vectors associated with the functionals. The Poisson bracket, unlike the continuous case, is defined in terms of finite differences:
\begin{definition} The Poisson bracket in the finite case is given by
$$[f,g]_{PB}(\mathbf{m})=\sum_{i=1}^n[(f(\mathbf{m}+\mathbf{q}_i)-f(\mathbf{m}))(g(\mathbf{m}+\mathbf{p}_i)-g(\mathbf{m}))-(f(\mathbf{m}+\mathbf{p}_i)-f(\mathbf{m}))(g(\mathbf{m}+\mathbf{q}_i)-g(\mathbf{m}))],$$
where the operations are in modulo $d$. The Poisson bracket, $[f,g]_{PB}(\mathbf{m})$, is also equal to symplectic inner product $\langle \mathbf{f}, \mathbf{g} \rangle$ on the discrete phase space.
\end{definition}

Like in the continuous case, an epistemic state is determined by the set of quadrature variables that are known to that agent and the values of these variables. This corresponds to the pair $(V^*,\mathbf{v})$, where $V^*$ is an isotropic subspace of the phase space $\Omega^*$ and  $\mathbf{v}$ is a valuation vector in $V^{**}=V$. Similarly, the valid transformations are symplectic transformations which preserve the symplectic inner product and they form the affine symplectic group over the finite field $\mathbb{Z}_d$. Note that these transformations over a finite field are discrete in time; hence, they cannot be generated from a Hamiltonian unlike the continuous case.

\begin{example}
As an example, we consider the quadrature epistricted theory of trits \cite{S} for a single system. The configuration space and the phase space are $\mathbb{Z}_3$ and $\mathbb{Z}_3^{2}$, respectively. The quadrature functionals in this system are of the form $f=\mathbf{a}q+\mathbf{b}p+\mathbf{c}$ where $\mathbf{a}$, $\mathbf{b}$, $\mathbf{c}\in \mathbb{Z}_3$. There are four inequivalent quadrature functionals:
$$q,p,,q+p, q+2p.$$
Since none of these functionals Poisson-commute, an agent can know at most one of them. This implies that there are twelve epistemic state as the valuation vectors are chosen from $V=\mathbb{Z}_3$. These states can be depicted as the following $3\times 3$ grids:
\begin{center}
\includegraphics[scale=0.5]{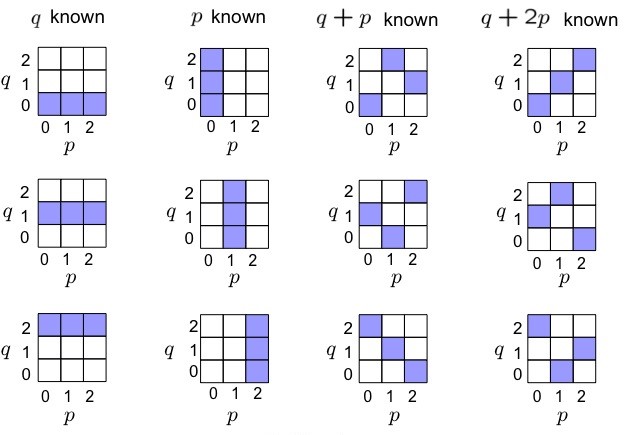}
\end{center}

The valid transformations, which forms affine symplectic group over $\mathbb{Z}_3$, correspond to a certain subset of permutations of the functionals. Here is an example:
\begin{center}
\includegraphics[scale=0.5]{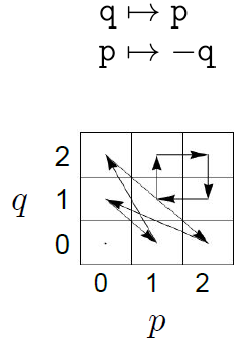}
\end{center}

\end{example}

\subsection{The Category of epistricted theories}
We now turn to the category of the epistricted theory of trits. The arguments can easily be generalized to the epistricted theories for other odd primes.We start with the category of \textbf{FRel} whose objects are sets and whose morphisms $X\to Y$ are relations $r\subseteq X \times Y$ and $s\circ r =\{(x,z)|\exists y, (x,y)\in r,(y,z)\in s\}$. \textbf{FRel} is a dagger symmetric monoidal category when the tensor product is chosen as a cartesian product, the single element set $1=\{\bullet \}$ as the identity and the relational converse as the dagger morphism $\dagger$.
\begin{definition}
An object $X$ in \textbf{FRel} with a morphism $m:X \times X \to X$ is called \emph{special dagger Frobenius algebra} if and only if $m$ has the following properties:
\begin{itemize}
\item $(1 \times m)\circ (m^{\dagger}\times 1)=m^{\dagger}\circ m=(m\times 1)\circ (1\times m^{\dagger})$ (F)
\item $m\circ m^{\dagger}=1$ (M)
\item $m\circ(1 \times m)=m\circ (m\times 1)$ (A)
\item there is $e:1\to X$ with $m\circ (e\times 1)=1=m\circ (1\times e)$ (U).
\end{itemize}
\end{definition}
The conditions of Frobenius algebras presented graphically in Figure 1. These diagrams encode composition by drawing morphisms on top of each other, and monoidal product is drawing morphism next to each other. The dagger is a vertical reflection.

\begin{figure}[h]
   \centering
   \def\svgwidth{150 pt}
  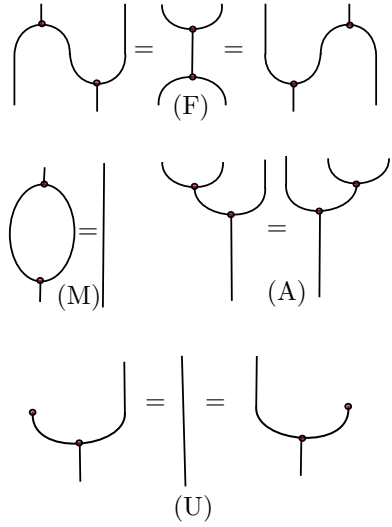
  \caption{String diagrams of the properties for the objects in \textbf{FRel}}
    \label{fig:Hol}
\end{figure}

The category \textbf{FRel} has morphisms $\eta:1 \to X \times X$ satisfying
\begin{itemize}
\item $(\eta^{\dagger}\times 1)\circ (1 \times \eta)=1=(1\times \eta^{\dagger}\circ (\eta \times 1))$(C).
\end{itemize}
\begin{proposition}(\cite{CCH}). \textbf{FRel} is a \emph{compact closed category}. 
 \end{proposition}
 The compact structure can be induced from the Frobenius algebra by $\eta=m^{\dagger}\circ e$. As a result of compact structure, we can define transposes of morphism $r:X \to Y $ by $\lceil r \rceil =(1 \times r)\circ \eta : 1\to X \times Y.$ The category of Frobenius algebras in \textbf{FRel} with the following morphism is a well-defined category (see Proposition 14 \cite{CCH}).

\begin{definition}
A morphism $(X,m_X)\to (Y, m_Y)$ in the category of Frobenius algebras in \textbf{FRel}is a morphism $r:X \to Y$ satisfying
\begin{itemize}
\item $(m_X \times m_Y)\circ (1 \times \sigma \times 1)\circ (\lceil r \rceil \times \lceil r \rceil)=\lceil r \rceil$
\item $(r \times \eta^{\dagger})\circ (m_X^{\dagger}\times 1)\circ (e_X \times 1)= (e_Y^{\dagger}\times 1)\circ (m_Y \times 1)\circ (r \times \eta)$
where $\sigma:X\times Y \to Y\times X$ is a natural swap isomorphism.
\end{itemize}
\end{definition}

\begin{figure}[h]
    \centering
    \def\svgwidth{150 pt}
    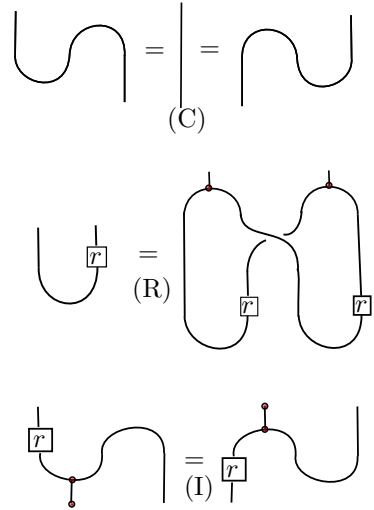
  \vspace*{-.5 cm}  \caption{String diagram of the properties for the morphisms in \textbf{Frel}}
    \label{fig:Hol}
\end{figure}

\begin{proposition}
The category \textbf{Spek} for the toy theory of trits is a subcategory of \textbf{FRel}.
\end{proposition}
\begin{proof}
The category \textbf{Spek} for the toy theory of trits is defined as the category  whose objects are the single element $1$ and $n$-fold cartesian product of nine element set $IX:=\{1,2,\ldots, 9\}$. The morphisms of \textbf{Spek} can be constructed by composition, Cartesian product and relational converse from the following relations:
\begin{itemize}
\item The unit(deleting) relation $e:IX \to 1$ defined by $\{1,4,7\}\sim \bullet$
\item The relation $m: IX \to IX \times IX$ defined as:
\newline
\begin{center}
\begin{tabular}{|c|c|c|c|c|c|c|c|c|}
  \hline
  % after \\: \hline or \cline{col1-col2} \cline{col3-col4} ...
  1 & 2 & 3 &   &   &   &   &   &   \\ \hline
  3 & 1 & 2 &   &   &   &   &   &    \\ \hline
  2 & 3 & 1 &   &   &   &   &   &   \\ \hline
    &   &  & 4 & 5 & 6 &   &   &   \\ \hline
    &   &   & 6 & 4 & 5 &   &   &    \\ \hline
    &   &   & 5 & 6 & 4 &   &   &    \\ \hline
    &   &   &   &   &   & 7 & 8 & 9  \\ \hline
    &   &   &   &   &   & 9 & 7 & 8  \\ \hline
    &   &   &   &   &   & 8 & 9 & 7 \\
  \hline
\end{tabular}
\end{center}
For example, $\{1\}\sim \{(1,1), (2,2), (3,3)\}$, $\{2\}\sim \{(1,2), (2,3), (3,1)\}$ etc.
\item The permutations $\sigma_i:IX \to IX$ that correspond to affine symplectic maps on the phase-space.
\item The relevant unit, associativity and symmetry natural isomorphisms.
\end{itemize}

Twelve epistemic states for a single system is given by the following relations:
\begin{center}
$q$ known: $\bullet \sim \{1,2,3\}$, $\bullet \sim \{4,5,6\}$, $\bullet \sim \{7,8,9\}$.

$p$ known: $\bullet \sim \{1,4,7\}$, $\bullet \sim \{2,5,8\}$, $\bullet \sim \{3,6,9\}$.

$p+q$ known: $\bullet \sim \{1,6,8\}$, $\bullet \sim \{2,4,9\}$, $\bullet \sim \{3,5,7\}$.

$p+2q$ known: $\bullet \sim \{1,5,9\}$, $\bullet \sim \{2,6,7\}$, $\bullet \sim \{3,4,6\}$.
\end{center}

It is straightforward to verify that $(IX,m,e)$ is the special dagger Frobenius algebra. 
\end{proof}
\begin{remark}
This structure correspond to the observable for which $q$ is known. Hence, the relations  $\bullet \sim \{1,2,3\}$, $\bullet \sim \{4,5,6\}$, $\bullet \sim \{7,8,9\}$ are the \emph{copyable (classical)} states for this observable. The other observables can be found by composing $m$ with various valid permutations.
\end{remark}

\subsection{Frobenius Algebras as Groupoids}

We start our procedure of the discrete geometric quantization with constructing the groupoid corresponding to the Frobenius algebra $(IX,m,e)$. The groupoid characterization of dagger Frobenius algebras is given in \cite{CCH}. We now give the groupoid following \cite{CCH}.

\begin{definition}
The following objects and morphisms in \textbf{Rel} obtained from the Frobenius algebra $(IX,m,e)$ form a groupoid $\Sigma$ in the category of sets and functions \textbf{Set} (see \cite{CCH} Theorem 7).
\begin{itemize}
\item $\Sigma_1= IX$

\item $\Sigma_2=Image(m)=\bigcup_{k=0}^2\{(3k+1,3k+1),(3k+1,3k+2),(3k+1,3k+3),(3k+2,3k+3),(3k+2,3k+2),(3k+2,3k+1),(3k+3,3k+1),(3k+3,3k+2),(3k+3,3k+3)\}$

\item $\Sigma_0=U=Domain(e)=\{1,4,7\}$

\item $u=U\times U: \Sigma_0 \to \Sigma_1$

\item $s=\{(f,x)\in \Sigma_1\times \Sigma_0 | (f,x) \in \Sigma_2\}:\Sigma_1 \to \Sigma_0$

\item $t=\{(f,y)\in \Sigma_1\times \Sigma_0 | (y,f) \in \Sigma_2\}: \Sigma_1 \to \Sigma_0$

\item $-^*=\{(g,f)\in \Sigma_2| m(g,f)\in U , m(f,g)\in U\}: \Sigma_1 \to \Sigma_1$
\end{itemize}
\end{definition}

\begin{remark}
As proven in \cite{CCH}, this assignment is functorial., if we consider morphisms of groupoids to be subgroupoids.
\end{remark}
Considering the set $IX$ as the finite field $\mathbb{Z}_3^2$, one can equip $IX$ with the symplectic product $\omega=\langle \cdot,\cdot \rangle$.

\begin{lemma}
The graph of the multiplication $\mathfrak{m}=\{(xy, x, x|(x,y)\in \Sigma_2)\}$ is a Lagrangian subspace of $\Sigma_2 \oplus \Sigma_2 \oplus \Sigma_2$.
\end{lemma}

\begin{proof}
$$\mathfrak{m}=\bigcup_{k=0}^2\{(3k+1,3k+1,3k+1),(3k+2,3k+1,3k+2),(3k+3,3k+1,3k+3),(3k+3,3k+2,3k+3),$$ $$(3k+1,3k+2,3k+2),(3k+2,3k+2,3k+1),(3k+2,3k+3,3k+1),
(3k+3,3k+3,3k+2),(3k+1,3k+3,3k+3)\}.$$

Equipped with the symplectic product $\omega=\langle \cdot,\cdot \rangle$,$\mathfrak{m}$ becomes Lagrangian subspace of $\Sigma_2 \oplus \Sigma_2 \oplus \Sigma_2=\mathbb{Z}_3^6$ with the basis $$\{((0,0)(0,1),(0,1)),((0,1),(0,0),(0,1)),((1,0),(1,0),(1,0))\}$$ where $(a,b) \in \mathbb{Z}_3^2$.
\end{proof}
\subsection{Weyl Correspondence and Pair Groupoid}\label{Weyl}

In order to apply Hawkins' quantization, we require a symplectic groupoid and a Lagrangian polarization. 

\begin{proposition}
In the discrete geometric quantization procedure for the symplectic space $(M=\mathbb{Z}_3^2,\omega)$, its pair groupoid and the subspace $P=\mathbb{Z}_3$ are suitable symplectic groupoid and groupoid polarization respectively.
\end{proposition}
\begin{proof}

The \emph{pair groupoid} associated to a set $M$ consists of  $Pair(M)= M \times M$, endowed with the multiplication $(x, y) \cdot (y, z) = (x, z)$. $M$ embeds in $M\times M$ as the diagonal $\{(x,x)| x\in M\}$, and $s$ and $t$ are the projections $s(x, y) = (y, y)$ and $t(x, y) = (x, x)$. In this groupoid, there is exactly one arrow from any object to another.
Starting with pair groupoid $M=\mathbb{Z}_3^2 \times \mathbb{Z}_3^2$, one can define a symplectomorphism $\Phi:\mathbb{Z}_3^2 \times \mathbb{Z}_3^2 \to \mathbb{Z}_3^2 \times \mathbb{Z}_3^{2*}=T^*(\mathbb{Z}_3^2)$. It is then clear that such pair groupoid is symplectic and it integrates the symplectic space $(M=\mathbb{Z}_3^2,\omega)$.\end{proof}

Now, choosing the polarization as
$$P=span\{p_1-p_2, q_1-q_2\},$$
we can adapt the procedure in continous case to obtain the discrete Fourier transform:
$$\mathcal{F}f(p,q)=C\sum_{b\in \mathbb{Z}_3}f(\frac{p+q}{2},b)e^{\frac{2\pi}{d}ib(p-q)}$$
which maps complex functions $f:\mathbb{Z}_3^2  \to \mathbb{C}$ to projective operators in the twisted group algebra $C^*(\mathbb{Z}_3^2, \sigma)$ as a result of Hawkins quantization.

\begin{remark}
Note that we cannot apply the same procedure to the toy bits, i.e. $\Omega=\mathbb{Z}_2$, as the symplectomorphism $\Phi$ and other steps of quantization include division by two.
\end{remark}
Our main result produces a functorial quantization via symplectic groupoids, in the case of epistricted theories with an odd prime number of degrees of freedom.
\begin{theorem}[Main result for the finite case]
The discrete geometric quantization procedure is a functor from the Frobenius algebra in \textbf{Rel} for epistricted theories to the Frobenius algebra for stabilizer quantum mechanics in odd prime discrete case.
\end{theorem}
\begin{proof}

Let us construct a pair groupoid $M$ from the dagger Frobenius algebra $(IX,m,e)$. We start with the monoid structure $(IX\times IX,id_{IX}\times \eta^{\dagger} \times id_{IX}, \eta)$ in \textbf{Rel}, where $\eta:=m^{\dagger} \circ e$. This monoid is a specific  example of \emph{endomorphism monoids} in \cite{V} which is an analogue of algebras of bounded linear operators. Note that the new monoid multiplication $m'=id_{IX}\times \eta^{\dagger} \times id_{IX}$ is precisely the multiplication in $m'((x, y), (y, z)) = (x, z)$ in the pair groupoid, and the unit is the diagonal $\eta=e \circ m:\bullet \to \{(a,a)|a\in IX\}$. The \emph{abstract polarization} $P$ in this context can be cast as $\bullet \to U\times U$. We denote this monoid as $End(IX)$

The algebra $(IX,m,e)$ can be embedded into endomorphism monoid $End(IX)$ similar to the fact that every algebra has a homomorphism into the algebra of operators. The embedding homomorphism $h:(IX,m,e)\to End(IX)$ is defined by:
$$h:=m.$$
It is easy to show that $h$ is an preserves multiplication and the unit. One can also refer to Lemma 3.19 in \cite{V} for a more general case. Let $(IX\times IX, \bar{m},\bar{e}=\eta)$ denote the image of $h$ in the endomorphism monoid. We now can construct the groupoid $\bar{\Sigma}$ from the dagger Frobenius algebra $(IX\times IX, \bar{m},\bar{e})$ following the construction in \cite{CCH} one more time:
\begin{itemize}
\item $\bar{\Sigma}_1=IX\times IX$
\item $\bar{\Sigma}_2=Image(\bar{m})$
\item $\bar{\Sigma}_0=\bar{U}=Image(\bar{e})$
\item $\bar{u}=\bar{U}\times \bar{U}$
\item $\bar{s}=\{(f,x)\in \bar{\Sigma}_1\times \bar{\Sigma}_0 | (f,x) \in \bar{\Sigma}_2\}:\bar{\Sigma}_1 \to \bar{\Sigma}_0$
\item $\bar{t}=\{(f,y)\in \bar{\Sigma}_1\times \bar{\Sigma}_0 | (y,f) \in \bar{\Sigma}_2\}: \bar{\Sigma}_1 \to \bar{\Sigma}_0$
\item $-^*=\{(g,f)\in \bar{\Sigma}_2| m(g,f)\in \bar{U} , m(f,g)\in \bar{U}\}: \bar{\Sigma}_1 \to \bar{\Sigma}_1$
\end{itemize}

By  Lemma 1 $\bar{\Sigma}$ can be equipped with a symplectic structure so that it becomes symplectic groupoid where the polarization is $P=span\{p_1-p_2, q_1-q_2\}$ corresponding to $\bullet \to U\times U$ in $(IX\times IX, \bar{m},\bar{e})$. Hence, the quantization give us a  subalgebra of $C^*(\mathbb{Z}_3^2, \sigma)$ as we only consider the linear combination of position and momentum operators. The resulting operator algebra is projective representation of finite Heisenberg group given by the above discrete Fourier transform. The resulting finite dimensional $C^*$-algebra is equivalent to a dagger Frobenius algebra in \textbf{Hilb} (see Theorem 4.7 \cite{V}). By the functoriality of quantization in this specific case and functoriality of the above embedding into $End(IX)$ (see corollary 4.4 \cite{CCK}, we obtain a functor from the dagger Frobenius algebras in \textbf{Rel} to the dagger Frobenius algebras in \textbf{Hilb}. The affine symplectic transformations of the epistricted theories are mapped the group representations of affine symplectic group which acts as a superoperator in the resulting $C^*$-algebra.
\end{proof}

%(cp*  prop 5.4)

\section{Conclusion and further work}
We have established the relationship between geometric quantization and quadrature subtheories for the continous degrees of freedom. We conclude that the group algebra $C^*(H)$ for Heisenberg group $H$ contains the quadrature subtheories as a result of groupoid quantization procedure. One can use this fact to give operator algebraic approach to quantum optics.

\subsection{$C^*$-quantization.}
This construction also suggests that there is a "geometric quantization" functor, from a subcategory of the category of groupoids to the category of $C^*$-algebras. Following \cite{CCH}, this corresponds to a functor from Frobenius algebras in the category \textbf{FRel} (Frobenius algebras in the category of sets and relations) to Frobenius algebras in the category of Hilbert spaces \textbf{FHilb}. The functor has to be defined in the subcategory of Frobenius algebras arising from symplectic groupoids, and the morphisms have to be adapted in order to obtain functoriality.
\subsection{The even case.}
We investigate discrete degrees of freedom. The variables in this case are chosen from a finite field instead of real numbers. Even though Spekkens' original toy theory \cite{S} is contained in the case where finite field is $\mathbb{Z}/2$, we consider odd degrees freedom. The reason is that for $\Omega= (\mathbb{Z}/2)^n$ the discrete Wigner representation can take negative values and therefore the epistricted theory does not coincide with the quadrature subtheories \cite{S1}. Our main result is to give a discrete version of groupoid quantization. The resulting algebra is $C^*(H)$ for the finite Heisenberg group $H$. This finite $C^*$-algebra corresponds to a Frobenius structure via the construction of Vicary \cite{V}. Thus, one can study quantum phenomena such as complementarity in quadrature theories in this algebraic framework.
\subsection{Geometric quantization over finite fields.} In the work of Gurevich and Hadani \cite {GH}, a functorial description of geometric quantization is developped for vector spaces over fields with positive characteristics. The odd prime case is resemblant to the discrete geometric quantization procedure we have described in this paper. We expect to have a more explicit  comparison in the future between our quantization procedure for the odd finite case and this geometric quantization program 

\section{Acknowledgements}
We thank to Chris Heunen for bringing the question of complementarity in Spekkens' toy theory to our attention. AND acknowledges King Fahd University of Petroleum and Minerals (KFUPM) for funding this work through project No. SR141007. IC was partially supported by the SNF grant PBZHP2-147294.

\nocite{*}
\bibliographystyle{eptcs}

\begin{thebibliography}{50}


\bibitem{AC}
S. ~Abramsky, B. ~Coecke,
\newblock \emph{A categorical semantics of quantum protocols},
\newblock In: Proceedings of the 19th Annual IEEE Symposium on Logic on Computer Science (LICS'04), 415-425.

\bibitem{AE}
S. T. ~Ali, T. ~Rudolph and M. Englis,
\newblock \emph{Quantization methods: a guide for physicist and analyst},
\newblock Rev. Math. Phys. \textbf{17} (2005), no. 4, 391-400.

\bibitem{BKK}
C. ~Beny, A. ~Kempf and D. W. ~Kribs,
\newblock \emph{Generalization of Quantum Error Correction via Heisenberg Picture},
\newblock Physical Review Letters \textbf{98} (2007), 100502.

\bibitem{BW}
S. ~Bates, A. ~Weinstein,
\newblock \emph{Lectures on Geometry of Quantization},
\newblock Berkeley Mathematics Lecture Notes, Vol. 8,(American Mathematical Society, Providence, RI, 1997).

\bibitem{B}
S. D. ~Bartlett, T. ~Rudolph and R. W. Spekkens,
\newblock \emph{Reconstruction of Gaussian quantum mechanics from Liouville mechanics with an epistemic restriction},
\newblock Physical Review A \textbf{86} (2012), no. 1, 012103.

\bibitem{CCH}
A.~Cattaneo, I. ~Contreras and C. ~Heunen,
\newblock \emph{Relative Frobenius Algebras are Groupoids},
\newblock Journal of Pure and Applied Algebra (2013), 217:114-124.

\bibitem{CCK}
B.~Coecke, C.~Heunen, A.~Kissinger
\newblock \emph{Categories of Quantum and Classical Channels},
\newblock Quantum Information Processing(2014), doi:10.1007/s11128-014-0837-4.


\bibitem{CE}
B.~Coecke and B. ~Edwards,
\newblock \emph{Toy quantum categories}, arXiv:1108.1978v1, 2011.
\newblock Electronic Notes in Theoretical Computer Science \textbf{270} (2011), no. 1, 29--40.


\bibitem{CES}
B.~Coecke, B. ~Edwards and R. W. Spekkens,
\newblock \emph{Phase Groups and the Origin of Non-locality for Qubits},
\newblock Electronic Notes in Theoretical Computer Science \textbf{270} (2011), no. 2, 15--36.

\bibitem{CHK}
B.~Coecke, C. ~Heunen and A.~Kissinger,
\newblock \emph{Categories of quantum and classical channels},
\newblock Electronic Proccedings in Theoretical Computer Science \textbf{158} (2014), 1--14.




\bibitem{G}
D.~Gross,
\newblock \emph{Hudson's Theorem for finite-dimensional quantum systems},
\newblock Journal of Mathematical Physics \textbf{47} (2006), 122107.

\bibitem{GH}
S.~Gurevich and R. ~Hadani,
\newblock \emph{Quantization of symplectic vector spaces over finite fields},
\newblock Journal of Symplectic Geometry \textbf{7} (2009), no. 4, 475--502.

\bibitem{GB-V}
J.M. ~Gracia-Bondia and J.C. Varilly,
\newblock \emph{From Geometric Quantization to Moyal Quantization},
\newblock Journal of Mathematical Physics \textbf{36}(2691)(1995),
\newblock arxiv: 9406170

\bibitem{H}
E.~Hawkins,
\newblock \emph{A groupoid approach to quantization},
\newblock Journal of Symplectic Geometry \textbf{6} (2008), no. 1, 61--125.

\bibitem{KLPL}
D. W. ~Kribs, R. ~Laflamme, D. ~Poulin and M. ~Lesosky
\newblock \emph{Operator Quantum Error Correction},
\newblock Quantum Information \& Computation \textbf{6} (2006), issue 4, 382--399.

\bibitem{M}
S.~Maclane,
\newblock \emph{Categories for the Working Mathemaician},
\newblock Second Edition (Springer--Verlag, 2000).

\bibitem{P}
D. ~Poulin,
\newblock \emph{Stabilizer Formalism for Operator Quantum Error Correction},
\newblock Physical Review Letters \textbf{95} (2005), 230504.

\bibitem{R}
M. A. ~Rieffel,
\newblock \emph{Quantization and $C*$-algebras},
\newblock Contemporary Mathematics \textbf{167} (1994), 67--97.

\bibitem{R2}
M. A. ~Rieffel
\newblock \emph{Deformation quantization for actions of $\mathbb{R}^d$},
\newblock Mem. Am. Math. Soc. \textbf{106}(506) (1993).

\bibitem{S1}
R. W. ~Spekkens,
\newblock \emph{Evidence for the epistemic view of quantum states: A toy theory},
\newblock Physical Review A \textbf{75} (2007), no. 3, 032110.


\bibitem{S}
R. W. ~Spekkens,
\newblock \emph{Quasi-quantization: classical statistical theories with an epistemic restriction},
\newblock arxiv: 1409.5041, 2014.

\bibitem{V}
J. ~Vicary,
\newblock \emph{Categorical Formulation of Finite-Dimensional Quantum Algebras},
\newblock Communications in Mathematical Physics \textbf{304} (2011), no. 3, 765--796.



\bibitem{W}
N. M. J. ~Woodhouse,
\newblock \emph{Geometric Quantization},
\newblock (Oxford University Press, New York, 1992).






%\bibitem[5]{A}
%A.~Adem and R.J.~Milgram,
%\emph{Cohomology of finite groups},
%\newblock Second edition, Grundlehren der Mathematischen Wissenschaften, 309. Springer-Verlag, Berlin, 2004.
%
%\bibitem[6]{BK}
%A.K.~Bousfield and D.M.~Kan,
%\newblock \emph{Homotopy limits, completions and localizations},
%\newblock Lecture Notes in Mathematics, Vol. 304 (Springer--Verlag, Berlin--New York, 1972).

\end{thebibliography}

\end{document}